\documentclass[11pt]{article}

\usepackage[margin=1in]{geometry}
\usepackage{CJK, hyperref, amsmath, amssymb, amsthm}

\newtheorem{theorem}{Theorem}
\newtheorem{proposition}{Proposition}
\newtheorem{lemma}{Lemma}

\theoremstyle{definition}
\newtheorem{definition}{Definition}

\theoremstyle{remark}
\newtheorem*{remark}{Remark}

\DeclareMathOperator{\tr}{tr}
\DeclareMathOperator{\poly}{poly}

\begin{document}

\begin{CJK*}{UTF8}{}

\title{Is microcanonical ensemble stable?}

\CJKfamily{gbsn}

\author{Yichen Huang (黄溢辰)\thanks{We acknowledge funding provided by the Institute for Quantum Information and Matter, an NSF Physics Frontiers Center (NSF Grant PHY-1125565) with support of the Gordon and Betty Moore Foundation (GBMF-2644).}\\
California Institute of Technology, Pasadena, California 91125, USA\\
ychuang@caltech.edu
}

\maketitle

\end{CJK*}

\begin{abstract}

No, in a rigorous sense specified below.

\end{abstract}

\section{Introduction}

For the purpose of this work, it suffices to work with a chain of $n$ spins (qu\emph{d}its), each of which has local dimension $d=\Theta(1)$. We are given a local Hamiltonian $H=\sum_{j=1}^{n-1}H_j$ with open boundary conditions, where $\|H_j\|=O(1)$ acts on the spins $j$ and $j+1$ (nearest-neighbor interaction). Since the standard bra-ket notation can be cumbersome, in most but not all cases quantum states and their inner products are simply denoted by $\psi,\phi,\ldots$ and $\langle\psi,\phi\rangle$, respectively, cf. $\||\psi\rangle-|\phi\rangle\|$ versus $\|\psi-\phi\|$. Let $\psi_1,\psi_2,\ldots,\psi_{d^n}$ be the eigenstates of $H$ with the corresponding eigenvalues $E_1\le E_2\le\cdots\le E_{d^n}$ in non-descending order. The projector onto the energy window $[E-\delta,E+\delta]$ is given by
\begin{equation}
P(E,\delta)=\sum_{j:|E_j-E|\le\delta}|\psi_j\rangle\langle\psi_j|.
\end{equation}

A microcanonical ensemble is a fundamental concept in statistical mechanics. Throughout this paper, we only consider the physical situation that the bandwidth is (at most) a constant.

\begin{definition} [microcanonical ensemble]
An (exact) microcanonical ensemble of energy $E$ and bandwidth $2\Delta_e=O(1)$ is the set
\begin{equation} \label{ext}
EXT=\{\psi:\psi=P(E,\Delta_e)\psi\}.
\end{equation}
\end{definition}

The state in practice may well only be approximately rather than exactly in a microcanonical ensemble. A state is in an approximate microcanonical ensemble if the population ``leakage'' outside a distance (in the spectrum) from the target energy is exponentially small in the distance.

\begin{definition} [approximate microcanonical ensemble]
An approximate microcanonical ensemble of energy $E$ and bandwidth $2\Delta_a=O(1)$ is the set
\begin{equation} \label{apx}
APX=\{\phi:|\langle\phi,P(E,x)\phi\rangle|\ge1-O(e^{-x/\Delta_a}),\forall x\ge0\}.
\end{equation}
\end{definition}

The stability of a microcanonical ensemble can be phrased as follows. Suppose a microcanonical ensemble has a universal physical property in the mathematical sense of an inequality satisfied by all states in $EXT$. Is this inequality valid (possibly up to small corrections) for all states in $APX$? If not, the physical property of the microcanonical ensemble is not robust against perturbations.

One might tend to believe that a microcanonical ensemble is stable due to a continuity argument. Given $\phi\in APX$, let $\psi=P(E,\Delta_e)\phi/\|P(E,\Delta_e)\phi\|$ so that $\psi\in EXT$ and $|\langle\psi,\phi\rangle|\ge1-O(e^{-\Delta_e/\Delta_a})$. For $\Delta_a\ll\Delta_e=O(1)$, the states $\psi,\phi$ are close to each other, and thus believed to behave similarly. The pitfall of this hand-waving argument is that $\psi,\phi$ differ only by a small constant, which has the potential of affecting the physics significantly.\footnote{For example, a generic state in a small-constant neighborhood of a product state has volume law for entanglement. The stability of area law for entanglement can be proved, but only if in the presence of additional structure.} Therefore, the continuity argument (if not combined with more sophisticated reasonings) does not immediately lead to the stability of a microcanonical ensemble.

We show that a microcanonical ensemble is unstable from an entanglement point of view.

\begin{definition} [entanglement entropy]
The Renyi entanglement entropy $R_\alpha(0<\alpha<1)$ of a bipartite (pure) quantum state $\rho_{AB}=|\psi\rangle\langle\psi|$ is defined as
\begin{equation}
R_\alpha(\psi)=(1-\alpha)^{-1}\log\tr\rho_A^\alpha,\quad\rho_A=\tr_B\rho_{AB},
\end{equation}
where $\rho_A$ is the reduced density matrix. The von Neumann entanglement entropy is defined as
\begin{equation}
S(\psi)=-\tr(\rho_A\log\rho_A)=\lim_{\alpha\rightarrow1^-}R_\alpha(\psi).
\end{equation}
\end{definition}

\begin{remark}
For fixed $\psi$, the Renyi entanglement entropy $R_\alpha$ is a non-increasing function of $\alpha$.
\end{remark}

We consider the evolution of entanglement entropy across a particular cut.

\begin{definition} [dynamical entanglement scaling exponent]
Suppose the state $\psi_0$ at time $t=0$ has bond dimension $D_0$ across the cut. Let $z$ be a nonnegative number such that
\begin{equation}
R_\alpha(e^{-iHt}\psi_0)\le\log D_0+O(t^z\poly\log t),~\forall t.
\end{equation}
\end{definition}

\begin{remark}
On the right-hand side, the first term is an upper bound on the entanglement of the initial state. Note that $D_0$ is allowed to grow (even exponentially, e.g., $D_0=d^{n/100}$) with the system size. The second term, which involves polylogarithmic corrections due to a technical reason, characterizes the growth of entanglement. 
\end{remark}

Traditional Lieb-Robinson techniques imply a universal bound $z\le1$ for arbitrary initial states. This bound can (cannot) be improved for states in an exact (approximate) microcanonical ensemble.

\begin{theorem} \label{diffusive}
For any initial state $\psi_0\in EXT$, we have $z\le1/2$, and this bound is tight.
\end{theorem}

\begin{proposition} \label{ballistic}
There is a Hamiltonian $H_{XX}$ and an initial state $\phi_0\in APX$ such that $z=1$.
\end{proposition}

\section*{Acknowledgments}

The author would like to thank John Preskill for an insightful comment.

\section{Proof of Theorem \ref{diffusive}}

We go beyond traditional Lieb-Robinson techniques using the idea of polynomial approximation. For the dynamics in a microcanonical ensemble, consider the Taylor expansion
\begin{equation}
e^{-iHt}\psi_0=\sum_{k=0}^{+\infty}\frac{(-iHt)^k}{k!}\psi_0\approx\sum_{k=0}^g\frac{(-iHt)^k}{k!}\psi_0,
\end{equation}
where $E=0$ is assumed without loss of generality. The truncation error is upper bounded by
\begin{equation} \label{error}
\sum_{k=g+1}^{+\infty}\left\|\frac{(-iHt)^k}{k!}\psi_0\right\|=\sum_{k=g+1}^{+\infty}\left\|\frac{(-iHt)^k}{k!}P(0,\Delta_e)\psi_0\right\|\le\sum_{k=g+1}^{+\infty}\frac{(\Delta_et)^k}{k!}\approx\frac{(e\Delta_et)^g}{g^g},
\end{equation}
which is super-exponentially small in $g$ for $g\ge3\Delta_et$. Let $\tilde O(x):=O(x\poly\log x)$ hide a polylogarithmic factor. A polynomial interpolation argument leads to the following result.

\begin{lemma} [\cite{AKLV13}, Lemma 4.2] \label{kitaev}
Suppose $\psi_0$ has bond dimension $D_0$ across a particular cut. The bond dimension of $p(H)\psi_0$ across the cut is $\le D_0e^{\tilde O(\sqrt g)}$, where $p$ is an arbitrary polynomial of degree $g$.
\end{lemma}

Combining Lemma \ref{kitaev} with the error estimate (\ref{error}), a straightforward calculation shows
\begin{equation}
R_\alpha(e^{-iHt}\psi_0)\le\log D_0+\tilde O(\sqrt{\Delta_et}+1/\alpha).
\end{equation}
Therefore, $z\le1/2$. To prove the tightness of this bound on $z$, it suffices to construct an example that violates the bound $z\le1/2-\delta$ for any $\delta>0$.

\begin{proposition} [\cite{CC05}]
Let $H_{\rm Is}$ be the Hamiltonian of the critical transverse-field Ising chain with length $n$, and $\psi_0$ be a product state that respects the $Z_2$ symmetry of $H_{\rm Is}$. The entanglement entropy $S(e^{-iH_{\rm Is}t}\psi_0)$ across the middle cut saturates to $\Omega(n)$ in time $t=O(n)$.
\end{proposition}

The Hamiltonian $H_{\rm Is}'=H_{\rm Is}/n$ has bandwidth $O(1)$. Hence, any state, including $\psi_0$, is in a microcanonical ensemble (with respect to $H_{\rm Is}'$). The entanglement entropy $S(e^{-iH_{\rm Is}'t}\psi_0)$ saturates to $\Omega(n)$ in time $t=O(n^2)$. This violates the bound $z\le1/2-\delta$.

\begin{remark}
To approximate the propagator with polynomials, we used the ``naive'' Taylor expansion, which is known to be non-optimal. The optimal approach is to expand $e^{-iHt}$ in the basis of the Chebyshev polynomials of the first kind. Unfortunately, this only improves the parameters hided in $\tilde O(\cdots)$. Also, the bound in Lemma \ref{kitaev} is tight up to polylogarithmic corrections due to the tightness of the bound $z\le1/2$.
\end{remark}

\section{Proof of Proposition \ref{ballistic}}
Consider the $XX$ chain of length $2n$ with a defect in the middle:
\begin{equation}
H_{XX}=(1-\lambda)(\sigma_n^x\sigma_{n+1}^x+\sigma_n^y\sigma_{n+1}^y)+\sqrt{1-\lambda^2}(\sigma_n^z-\sigma_{n+1}^z)-\sum_{j=1}^{2n-1}\left(\sigma_j^x\sigma_{j+1}^x+\sigma_j^y\sigma_{j+1}^y\right),~0\le\lambda\le1,
\end{equation}
where $\sigma_j^x,\sigma_j^y,\sigma_j^z$ are the Pauli matrices at the site $j$. Let $\phi_0=|\Uparrow\rangle\otimes|\Downarrow\rangle$ with $|\Uparrow\rangle=|\uparrow\rangle^{\otimes n}$ and $|\Downarrow\rangle=|\downarrow\rangle^{\otimes n}$. The entanglement entropy across the middle cut grows linearly with time only in the presence of a detect $\lambda\neq1$.

\begin{proposition} [\cite{EP12}] \label{linear}
In the thermodynamic limit, we have
\begin{equation}
S(e^{-iHt}\phi_0)=h(\lambda^2)t/(4\pi)+O(\log t),\quad h(x):=-x\ln x-(1-x)\ln(1-x).
\end{equation}
\end{proposition}

\begin{proposition} \label{inAPX}
The state $\phi_0$ is in an approximate microcanonical ensemble with $E=2\sqrt{1-\lambda^2}$ and $\Delta_a=20$.
\end{proposition}

\begin{proof}
We decompose $H_{XX}$ into three parts: $H_{XX}=H_L+H_\partial+H_R$, where $H_L, H_R$ include the terms acting only on the left or right half of the chain, and $H_\partial=-\lambda(\sigma_n^x\sigma_{n+1}^x+\sigma_n^y\sigma_{n+1}^y)$ is the term across the middle cut. Note that $H_L,H_R$ are decoupled from each other. For the domain wall state $\phi_0=|\Uparrow\rangle\otimes|\Downarrow\rangle$, it is easy to see that $|\Uparrow\rangle$ or $|\Downarrow\rangle$ is an eigenstate of $H_L$ or $H_R$ with energy $\sqrt{1-\lambda^2}$. The proof is completed by applying Theorem 2.3 in Ref. \cite{AKL16}.
\end{proof}

\end{document}